\documentclass[10pt]{amsart}

\usepackage{amsmath, amsthm, amsopn, amsfonts, verbatim}
\usepackage{amssymb}

\usepackage{epic}
\usepackage{eepic}
\usepackage{graphicx, color}
\usepackage{psfrag}

\title[Eigenvalue variations and semiclassical concentration]
{Eigenvalue variations and semiclassical concentration}
\author[L. Hillairet]{Luc Hillairet}

\address{Laboratoire de Math\'ematiques Jean Leray \\
UMR CNRS 6629-Universit\'{e} de Nantes, 2 rue de la Houssini\`{e}re, \\
BP 92 208, F-44 322 Nantes Cedex 3, France}
\email{Luc.Hillairet@math.univ-nantes.fr}

\newtheorem{thm}{Theorem}[section]

\newtheorem{coro}[thm]{Corollary}
\newtheorem{lem}[thm]{Lemma}
\newtheorem{prop}[thm]{Proposition}

\theoremstyle{definition}

\newtheorem{defn}[thm]{Definition}

\theoremstyle{remark}

\newtheorem{remk}[thm]{Remark}

\newcommand{\eps}{\varepsilon}

\newcommand{\alp}{\alpha}

\newcommand{\W}{{\mathcal{W}}}
\newcommand{\V}{{\mathcal{V}}}
\newcommand{\U}{{\mathcal{U}}}

\newcommand{\R}{{\mathbb{R}}}

\newcommand{\refeq}[1]{(\ref{#1})}
\newcommand{\und}{\frac{1}{2}}

\newcommand{\Id}{\mbox{Id}}
\newcommand{\op}{\mbox{Op}^+}
\newcommand{\supp}{\mbox{supp.}}

\begin{document}

\maketitle
\section{Introduction}
In this paper, we aim at expliciting some relations between 
the behaviour of analytic eigenvalue branches of 
a spectral problem and concentration properties of its 
eigenfunctions. In order to illustrate this we consider a 
semiclassical Schr\"odinger operator $-h^2 \Delta +V$ on $\R^d$ 
(in which $\Delta$ is the non-positive Euclidean Laplace operator). 
Viewing $h$ as an analytic parameter, this problem enters 
the usual perturbation theory and the eigenvalues 
are organised into analytic eigenbranches $E_j(h).$ 
It is then natural to address the behaviour of these eigenbranches 
when $h$ goes to $0$. This point of view 
is fundamentally different from the usual semiclassical perspective. 
Indeed, any analytic eigenbranch is expected to cross a typical 
energy window $[E_1,E_2]$ so that any semiclassical 
expansion in this window corresponds actually to different 
eigenbranches. Our point is that knowing how the  
eigenfunctions concentrate in the semiclassical regime yields
some result on the behaviour of eigenbranches. For instance, 
we have the following theorem (see section \ref{SecSchro}). 

\begin{thm}
Let $E_j(h)$ be an analytic eigenbranch of the family of operators 
$-h^2\Delta+V$ on $L^2(\R^d).$ The following then holds.
\begin{enumerate}
\item The eigenbranch $E_j(h)$ converges to some limit $E_0$ when $h$ tends to $0$. 
\item The limit $E_0$ is necessarily a critical value of the potential $V.$
\end{enumerate}
\end{thm}

We will also use the same kind of ideas to prove an 
integrated remainder estimate for an analytic variation of 
metrics on a smooth compact manifold $X.$ (see Thm. \ref{IREthm} 
in section \ref{SecIRE}). 

As we have already pointed out, the study of analytic eigenbranches 
and semiclassical analysis are two different worlds so that such 
a relation between them is {\em a priori} quite surprising. However, 
as it will become clear later, the formula that expresses 
the derivative of any analytic eigenbranch is 
actually a semiclassical quantity.  This is the fact we will be exploiting here. 
This idea is also used in the two recent papers \cite{HJ1} and \cite{Hass}. 
In \cite{HJ1} we deal with the limiting behaviour of 
analytic eigenbranches in a singular setting. 
The main ingredient of the proof consists in controlling certain 
quadratic forms evaluated on the eigenfunctions of the spectral problem 
under consideration. We claim here that such a control may actually 
be interpreted as a condition on the associated semiclassical measures. 
Such an interpretation is not necessarily 
helpful since rather little is known in general about 
semiclassical measures and the overall strategy of \cite{HJ1} 
relies more on direct estimates for solutions of one dimensional ODE's. 
In \cite{Hass}, a part of the proof relies on the fact that 
an assumption on semiclassical measures allows to prove an 
integrated remainder estimate of the form we will consider here. 
Section \ref{SecIRE} can actually be seen as a formalization 
of this argument of \cite{Hass}.    
\hfill \\

\noindent {\textsc{Acknowledgments :}} The ideas in this paper grew out from 
several discussions with Chris Judge and Andrew Hassell, I would like to thank both of them 
for their stimulating influence. I would also like to thank the CRM in Montr\'eal and 
the organizers of the workshop "Spectrum and Dynamics" in April, 2008 and 
the MSRI in Berkeley where part of this work was done. Finally, 
I would like to thank Gilles Carron, Yves Colin de Verdi\`ere, 
Patrick G\'erard and Andrew Hassell who made several comments on the first version 
of this paper that helped me improve the results and the exposition.

\section{quantization and Semiclassical measures}
In this section we will briefly recall some 
known facts about quantization and semiclassical measures. 
We will use both the classical and 
the semiclassical theory of pseudodifferential operators (See \cite{Tay, SjoDim} 
for background on these notions) 

In the case of Schr\"odinger operators in $\R^d$, the principal symbol of a 
($h$-)pseudo\-diffe\-ren\-tial operator is a function in $T^*(\R^d).$ When working on a 
smooth compact Riemannian manifold, we will denote by $S^*X$ the quotient 
of $T^*X \backslash \{ 0\}$ by $(x,\xi)\sim (x,\lambda \xi),~\lambda >0.$
The subsets of $S^*X$ are the conical subsets of $T^*X\backslash \{0 \}.$ 
The symbol of a pseudodifferential operator $A$
is defined by testing $A$ again fast oscillatory functions. It 
is then a function on $T^* M\backslash \{ 0\}$. The symbol of a classical 
pseudodifferential operator is conically invariant and thus defines 
a function on $S^*X.$  

\begin{remk}\label{rksymbol} 
The principal symbol of a pseudodifferential operator 
may also be read off from the Schwartz kernel of $A.$ Observe, however, that 
when performing this identification, one has to pay careful attention 
to the density that is used (especially when dealing with varying Riemannian 
metrics).
\end{remk}

A crucial point is the ability of quantizing any reasonable function 
on $S^*X.$ This is done via a so-called quantization procedure. There 
are several possible choices, but it is useful to choose a 
quantization that respects positivity (see \cite{CdV,HMR})\footnote{
It is actually critical that such a quantization exists}. 
Thus, for any  function $a$ on $S^*X$ 
there exists a pseudodifferential operator $\op(a)$ of order $0$ 
whose principal symbol is $a.$ If $A$ is any other pseudodifferential 
operator with principal symbol $a$ then $\op(a)-A$ is a 
pseudodifferential operator of order $-1.$ 
Moreover if, $a$ is non-negative, then $\op(a)$ is also non-negative, 
as an operator.  

In the semiclassical setting, the homogeneity condition 
on $a$ is dropped and we may quantize any smooth function 
on $T^*(X)$ with compact support. 

Using this quantization, we define the so-called semiclassical measures. 
We present the construction in the semiclassical setting since 
in the classical one we will use a non-concentration estimate rather than 
semiclassical measures. 

We start with a Schr\"odinger operator on $\R^d,$ and we 
consider a sequence of 
normalized eigenfunctions $u_h$ of the equation 
$$
-h^2 \Delta u_h+V(x)u_h\,=\,E_h u_h ,
$$
such that the corresponding eigenvalues $E_h$ converge to 
some non-critical energy $E_0.$ To any 
compactly supported function $a$ the map  
$a\mapsto \langle \op(a)u_h,u_h\rangle$ defines a measure 
that converges weakly (up to extracting a subsequence). Any such limit is 
called a semiclassical measure. 

\begin{remk}
One could in fact associate semiclassical measures to any bounded sequence 
in $L^2.$ In the sequel we will consider only semiclassical measures that 
are associated with eigenfunctions of the Laplace operator.
\end{remk}

Among the basic facts satisfied by any semiclassical measure, we mention 
only the two following properties (see \cite{CdV,HMR})
\begin{enumerate}  
\item Any semiclassical measure is a 
probability measure supported on the energy surface 
$\Sigma_{E_0}:=\{ |\xi|^2+V(x)\,=\,E_0\}.$   
\item Any semiclassical measure is invariant under the Hamiltonian flow 
associated with $|\xi|^2+V(x)$ on $\Sigma_{E_0}.$
\end{enumerate}

\begin{remk}\label{xi2}
Under suitable hypotheses, we may extend the definition 
to smooth functions that are not necessarily 
compactly supported in $T^*\R^d.$ In particular, 
for a sequence of eigenfunctions with energy $E_h$ converging to $E_0$, 
the quantity $\int \|\nabla u_h\|^2 dx $ is uniformly bounded. This allows to prove 
that $\int |\xi|^2 d\mu < + \infty$ and then that 
$$ \int |\xi|^2 d\mu \,=\, \lim_{h\rightarrow 0} \int \|\nabla u_h\|^2\,dx.
$$
\end{remk}
\hfill \\

In the classical setting we can also define semiclassical measures. 
The main difference is that instead of a small parameter 
going to zero, we have to consider a sequence of 
eigenfunctions with eigenvalues growing to infinity and we recover a probability 
measure on $S^*X.$ It can be pointed out that 
the estimate in corollary \ref{LowerCor} that we use in section \ref{SecIRE} is a kind of  
non-concentration estimate. Indeed, such an estimate would give some information on any semiclassical 
measure when passing to the high energy limit. The assumptions on semiclassical 
measure in section \ref{SecSchro} and on geometric control in section \ref{SecIRE} 
are thus very similar in nature.
 
\section{Schr\"odinger operators from the analytic viewpoint}\label{SecSchro}
In this section we consider the semiclassical Schr\"odinger equation in $X\,=\,\R^d$. 
The eigenvalue problem consists in looking for 
eigenfunctions $u_h$ and eigenvalues $E_h$ such that 
\begin{equation}
-h^2 \Delta u_h+V(x)u_h\,=\,E_h u_h ,
\end{equation} 
in which $\Delta$ is the (non-positive) Laplace operator 
associated with the Euclidean metric.
Let us assume that $V$ is smooth and confining (i.e. $\lim_{|x|\rightarrow \infty} V(x)=+\infty$) 
so that there is a complete orthonormal set of eigenfunctions 
and the corresponding eigenvalues grow to infinity. 

This eigenvalue problem may be put into a generalized eigenvalue problem 
in the sense of Kato (cf \cite{Kato}). We change notations by letting $h=t$ 
thus emphasizing that we will now consider an analytic perturbation problem. 
An integration by parts leads to the following equivalent 
problem 
\begin{equation}{\label{GenKato}}
q_t(u_t,v) = E_t \cdot n(u,v),
\end{equation} 
In which $q_t$ is the quadratic form 
$$q_t(u)\,=\,t^ 2 \int |\nabla u(x)|^2 dx+\int V(x)|u(x)|^2 dx$$ defined 
on $\mathcal{D}\,=\,H^1(X)\cap L^2( (1+V(x)^2)^{\und}dx)$ and $n$ is 
the standard Riemannian quadratic form on $L^2(X)$. 
Observe that although $q_t$ is real-analytic for $t\in \R,$ 
$\mathcal{D}$ is the form domain of $q_t$ only for $t\neq 0.$ 
Consequently, usual analytic perturbation theory tells us 
that the eigenvalues organize into real-analytic eigenbranches  
for $t \in (0,1].$ 
A standard question in this setting is to address the limiting 
behaviour of these eigenbranches when $t$ goes down to $0.$ 

\begin{thm}
Let $E_t$ be an eigenvalue branch of \refeq{GenKato}, 
then $E_t$ converges to a limit $E_0$ when $t$ goes to $0.$ 
Moreover, this limit $E_0$ is a critical value of $V.$   
\end{thm}

\begin{proof}
Let $E_t$ be an eigenvalue branch and $u_t$ the corresponding 
normalized eigenvector branch. The derivative of $E_t$ may be computed 
by differentiating eq. \refeq{GenKato} (see \cite{Kato}), this yields 
\begin{equation}\label{Deriv}
\dot{E_t}\,=\,2t \int_X |\nabla u_t|^2(x)dx.
\end{equation}
This expression is obviously non-negative so that $E_t$ is a non-decreasing 
function of $t.$ Since, for any $t$ we have $E_t\geq \min_{X} V(x),$ $E_t$ 
has a limit $E_0$ when $t$ goes to $0$.

The condition on concentration of semiclassical measures then enters 
in the following lemma.

\begin{lem}\label{SCM}
Assume that for any semiclassical measure $\mu$ at the energy $E_0,$ 
$$(C)~~~~~~~\int_{\Sigma_{E_0}} |\xi|^2 d\mu(x) \,>\,0.$$  
Then, for any eigenbranch $(E_t,u_t)$ such that $\lim E_t=E_0,$  
\begin{equation}\label{LowerBound}
 \exists c>0,~ \liminf_{t\rightarrow 0} ~~t^2\cdot \int_X |\nabla u_t(x)|^2 dx \geq c.
\end{equation}
\end{lem}   
Assuming this lemma, we finish the proof of the theorem. 
If we suppose, for a contradiction, that $E_0$ satisfies condition $(C)$, then 
using \refeq{Deriv}, and \refeq{LowerBound} we get 
$$
\dot{E}_t \geq \frac{2c}{t},
$$   
and this is a contradiction since $E_t$ has to converge when $t$ goes 
down to $0.$ So condition $(C)$ is not fulfilled and there exists a semiclassical measure 
at energy $E_0$ such that  
$$\int_{\Sigma_{E_0}} |\xi^2| d\mu \,=\,0.$$
This implies that the support of $\mu$ is contained in the level set 
$\{ \xi =0 \}\subset \Sigma_{E_0}.$ Since $d\mu$ has to be invariant by 
the hamiltonian flow 
$$
\left\{ \begin{array}{lcr}
\dot{x}&=& 2\xi \\
\dot{\xi}&=& -\nabla V(x)
\end{array}
\right.
$$   
this forces $E_0$ to be a critical value of $V$.
\end{proof}

It remains to prove lemma \ref{SCM}. 
\begin{proof}[Proof of Lemma \ref{SCM}] 
As it is standard in this 
kind of settings (see \cite{HHM} or \cite{BZ}), we actually work in the 
reverse direction, assuming that \refeq{LowerBound} is not true. We may thus 
find a sequence $t_n$ going to $0$ and corresponding eigenfunctions and 
eigenvalues $E_{t_n}$ and $u_{t_n},$ such that 
$$
t_n^2\int_X |\nabla u_{t_n}(x)|^2 dx \leq \frac{1}{n}.
$$
We now change notations again and let 
$t_n=h$ since the rest of the argument is of semiclassical 
nature (and, as usual in this setting $h$ actually stands for $h_n$). 
We thus have
\begin{equation}\label{Upper}
0 \leq h^2\int_X |\nabla u_{h}(x)|^2 dx \leq \eps(h).
\end{equation}
where $\eps(h)$ is some function going to $0$ with $h.$
The sequence $u_h$ is $L^2$ normalized so that we may extract subsequences 
and find associated semiclassical measures at energy $E_0$ since 
$\lim E_h=E_0$. Let $d\mu$ be one of 
these semiclassical measures, by definition (see remark \ref{xi2} above)
\begin{equation}\label{limit}
\lim_{h\rightarrow 0}~~h^2\cdot \int_X |\nabla u_{h}(x)|^2 dx \,=\,\int_{\Sigma_{E_0}} |\xi^2| d\mu, 
\end{equation}
where the limit is understood along the subsequence defining $\mu.$ 
Putting \refeq{Upper} and \refeq{limit} together yields 
$$
\int_{\Sigma_{E_0}} |\xi^2| d\mu \, \,=\,0,
$$
thus finishing the proof of the lemma. 
\end{proof}

\begin{remk}
Since the energy surface is defined by $|\xi|^2+V(x)=E_0$, we can replace 
everywhere 
$$ \int_{\Sigma_{E_0}} |\xi^2| d\mu $$ by 
$$\int_{\Sigma_{E_0}} (E_0-V(x)) d\mu.$$
\end{remk}

In dimension $1$ this result can be refined using that the spectrum of 
a one-dimensional Schr\"odinger operator is known to be simple. 
If we assume that the potential $V$ has only non-degenerate minima, then the 
spectrum near the bottom of the energy is known (see \cite{SjoDim}) 
and there is an infinite number of eigenvalues close to the minimum. Since the eigenvalue 
branches cannot cross, the non-degenerate minimum is then only possible limit.\footnote{
I would like to thank San V\~u Ng\d oc and Fr\'ed\'eric Faure for pointing out that, in the one-dimensional case, 
the simplicity of the spectrum allows to improve the result.}

\section{Integrated remainder estimates}\label{SecIRE} 
In this section, we consider a smooth compact manifold $X$ of dimension 
$d\geq 2$ and a real-analytic family of Riemannian metrics $(g_{\tau})_{\tau\in[-1,1]}.$
We denote by $\Delta_\tau$ the Laplace operator associated with $g_\tau$ and 
by $E_j(\tau)$ the associated analytic branches of eigenvalues. We will also 
denote by $\langle \cdot,\cdot\rangle_\tau$ the scalar product on $L^2(X,g_\tau).$
We now define the usual following counting functions :

\begin{eqnarray*}
N(\tau,E)&=&\sharp \{ j~|~E_j(\tau)\leq E \},\\
R_M(\tau,E)&=&\sharp \{ j~|~E_j(\tau)\in (E-M,E+M]\,\} \\
           &=& N(\tau,E+M)-N(\tau,E-M).
\end{eqnarray*}

\begin{remk}\label{uniform}
Observe that since $(g_\tau)$ is analytic in $\tau$, the scalar products 
$\langle \cdot,\cdot\rangle_\tau $ define equivalent norms. 
The same is true at the level of $H^1$ : the norms 
$\| \cdot\|_{H^1(X,g_\tau)}$ defined by :
$$ \| u\|_{H^1(X,g_\tau)}^2\,=\, \int_M g_\tau(\nabla_\tau u,\nabla_\tau u)dvol_{g_\tau} + \| u\|_\tau^2,$$ 
are all equivalent.
\end{remk}
 
The following lemma expresses the variation of $E_t$ using a family of 
pseudodifferential operators. It also serves as a definition of the family $A_\tau$ that 
will be used in the rest of the paper.

\begin{lem}
There exists a real-analytic family of symbols $a_\tau,$ and, for any $\tau_0,$ 
there exists a constant $C$ such that, for any  normalized eigenbranch 
$(E_n(\tau),u_n(\tau))$ we have 
$$
\forall |\tau|\leq \tau_0,~ \left | \frac{\dot{E}_n(\tau)}{E_n(\tau)} - \langle \mbox{Op}^+(a_\tau) u_n(\tau),u_n(\tau)\rangle_\tau\right | \leq C E_n(\tau)^{-\und}.
$$
We will denote by $A_\tau\,=\, \mbox{Op}^+(a_\tau).$ 
\end{lem}
\begin{proof}
First, by differentiating the eigenvalue equation 
$$ E_n(\tau)\,=\,\langle \Delta_\tau u_n(\tau),u_n(\tau)\rangle_\tau,$$
we find a second-order differential operator $\dot{\Delta}_\tau$ such that 
$$ \dot{E_n}(\tau)\,=\,\langle {\dot{\Delta}}_\tau u_n(\tau),u_n(\tau)\rangle_\tau.$$
(Observe that since the branch is normalized and using the eigenvalue equation we have that 
$\langle \Delta \dot{u},u\rangle = E\langle \dot{u},u\rangle = 0$).
Denote by $a_\tau$ the principal symbol of $\dot{\Delta}_\tau(\Delta_\tau +1)^{-1},$ and 
let $A_\tau=\mbox{Op}^+(a_\tau).$
By definition $R_\tau\,=\,\dot{\Delta}_\tau \left(\Delta_\tau +1\right)^{-1}-A_\tau$ is an analytic 
family of pseudodifferential operators of order $-1.$ 
In particular, $R_\tau$ is uniformly bounded from 
$L^2$ into $H^1.$ This yields the bound 
$$\left | \frac{\dot{E}_n(\tau)}{E_n(\tau)+1} - \langle \mbox{Op}^+
(a_\tau) u_n(\tau),u_n(\tau)\rangle_\tau\right | \leq C E_n(\tau)^{-\und}.$$
We thus get 
$$\left | \frac{\dot{E}_n(\tau)}{E_n(\tau)} - \langle \left (1+E_n^{-1}\right)
\mbox{Op}^+(a_\tau) u_n(\tau),u_n(\tau)\rangle_\tau\right | 
\leq C E_n(\tau)^{-\und}\left( 1+E_n^{-1}\right),$$
from which the claimed bound follows since 
$\left| \langle \mbox{Op}^+(a_\tau) u_n(\tau),u_n(\tau)\rangle_\tau \right |$ is uniformly bounded.
\end{proof} 

The main result of this section is then the following theorem (the notion of geometric control is defined 
in definition \ref{defgc} in the following section).

\begin{thm}\label{IREthm} 
Suppose that $a_0$ is non-negative and that there exists $\eps>0$ such that the subset 
$\left\{ (x,\xi) \in S^* X~,~a_0(x,\xi)\,>\eps \right \}$  geometrically controls $(S^*X,g_0)$, then 
for any fixed $M$, there exists $\tau_0,$ $E_0$ and $K$ such that 
\begin{equation}\label{IRE} 
\forall\, E>E_0,~~~\int_{-\tau_0}^{\tau_0} R_M(t,E) \,dt ~\leq ~K\, E^{\frac{d}{2}-1}.
\end{equation}
\end{thm}

The universal remainder in Weyl's law gives that for a fixed $t$, 
$R_M(t,E)=O(E^{\frac{d-1}{2}})$ as $E$ goes to infinity. The theorem 
thus says that averaging with respect to the perturbation 
greatly improves this estimate.
  
The rest of the section is devoted to the proof of this theorem. 
The central part of the argument is actually interesting by itself 
and consists in proving a somewhat explicit uniform control estimate. 

\subsection{A uniform control estimate}
We begin by recalling the notion of geometric control observing that since the 
velocity is constant along a geodesic, the geodesic flow 
of a Riemannian metric is well-defined on $S^*X.$ 

\begin{defn}\label{defgc}
We say that an open subset $\U$ of $S^*X$ \em{geometrically controls} 
$(S^*X,g)$ if for any $(x,\xi) \in S^*X$, there exists $T\in\R$ and 
$(x_0,\xi_0) \in \U$ such that $(x,\xi)\,=\,\Phi^T(x_0,\xi_0)$ where 
$\Phi^.(.,.)$ denotes the geodesic flow of $g.$
\end{defn}    

It is well-known that geometric control implies a non-concentration 
estimate for eigenfunctions. More precisely, if $U$ is an open subset of 
$X$ such that $\U=S^*U$ geometrically controls $(S^*X,g)$ then, there 
exists some positive $c$ and an energy $E_0$ such that
\begin{equation}\label{ControlEstimate}
\int_U |u_n|^2 dvol_g \geq c\int_X |u_n|^2 dvol_g,
\end{equation}
for any eigenfunction $u_n$ of the Laplace operator associated with $g$ such that $E_n\geq E_0$. 
(Actually the restriction on $E$ for this estimate may be released using the principle of unique 
continuation for solutions of second order elliptic PDE's see remark \ref{PUC} below)

Such a control estimate can be proved by a contradiction argument relying 
on known properties of semiclassical measures (\cite{HHM}, \cite{BZ} for the scheme of such an argument). 
Such an approach can probably be adapted to get a control estimate that is uniform in $\tau.$ 
We propose here a slightly different proof that also allows us to get some uniform control on $c.$

\begin{prop}\label{UniformProp}
Let $\U$ be an open subset of $S^*X$ that geometrically controls $(S^*X,g_0)$. 
There exists $\tau_0\in \R,$ $E_0\in \R,$ a zeroth-order pseudodifferential operator 
$\Pi$ on $X$ and $c>0$ such that : 
\begin{enumerate}
\item The symbol $\pi$ of $\Pi$ is supported in $\U$ and, $0 \leq \pi(x,\xi) \leq 1$ on $S^*X.$ 
\item For any $|\tau|\leq \tau_0$ and any $u$ eigenfunction of $\Delta_\tau$ with energy greater than $E_0$ 
then the following non-concentration estimate holds : 
\begin{equation}\label{UniformEst}
  \| \Pi u\|_\tau \geq c \| u\|_\tau.
\end{equation}
\end{enumerate}
\end{prop}
  
\begin{remk}\label{PUC}
Estimate \refeq{ControlEstimate} may be extended to any eigenfunction, without 
any restriction on the energy. Such an extension requires that one knows that an eigenfunction 
of a Laplace operator cannot vanish on an open subset (thus proving the control estimate 
for any finite number of eigenfunctions). The principle of unique continuation 
gives this property for eigenfunctions of the Laplace operator. 
It is not clear to the author if this non-vanishing 
property holds in the cotangent bundle. It can also be observed that, strictly speaking, 
the formulation with the microlocal cutoff only has a meaning in the high energy limit 
since the microlocal cutoff is defined modulo smoothing operators.  
\end{remk}
   
\begin{proof}
Let $(x,\xi)\in S^* X$, the geometric control assumption implies that there exists 
some $T(=T(x,\xi))$ and $(x_0,\xi_0)\in \U$ such that $\Phi_0^T(x_0,\xi_0)=(x,\xi).$ 
Using the continuity properties of solutions of an O.D.E. both with respect to 
initial conditions and with respect to parameters, we can find 
\begin{itemize}
\item[-]some $\tau_{\max}(x,\xi),$ 
\item[-] two conical neighbourhoods of $(x,\xi)$ : $\V=\V_{(x,\xi)}$ and $\W=\W_{(x,\xi)},$ 
\item[-] a conical neighbourhood $\tilde{\U}=\tilde{\U}_{(x,\xi)}$ of $(x_0,\xi_0)$ 
\end{itemize}
such that 
$\V$ is compactly included in $\W$ and $\Phi_\tau^{-T}(\W) \subset \tilde{\U} \Subset \U$ 
for any $|\tau| \leq \tau_{\max}(x,\xi).$ 
Since $S^*X$ is compact, we can find a finite collection $\left( (x_i,\xi_i)\right )_{1\leq i\leq N}$ 
such that $S^*X\,=\, \bigcup_{i=1}^N \V_{(x_i,\xi_i)}.$ 
We now use the index $i$ to denote 
any quantity formerly indexed by (or attached to) $(x_i,\xi_i).$ 
Denote by $\tau_0$ the minimum of all $\tau_{{\max}}(x_i,\xi_i).$
We can then find a smooth partition of unity $\pi_i$ such that 
each $\pi_i$ is identically $1$ in $\V_i$ and identically $0$ outside $\W_i.$ Note that, by construction, 
 
$$
\forall \, |\tau|\leq \tau_0,~~\bigcup_{i=1}^N \Phi_\tau^{-T_i}(\W_i)
\, \subset\, \bigcup_{i=1}^N \tilde{\U}_i \, \Subset\, \U.
$$  
We denote by $\pi_0$ a smooth function which is identically $1$ in 
$\bigcup_{i=1}^N \tilde{\U}_i$ and $0$ 
outside $\U.$ Denote by $\Pi_i\,=\op(\pi_i),$ any by 
$U_{\tau}(T)\,=\,\exp~iT\sqrt{\Delta_\tau},$ we claim that 
for any $i>0$, and for any $|\tau|\leq \tau_0,$ 
the operator $R_i\,=\,\Pi_i U_\tau(T_i) \left(\mbox{Id}-\Pi_0 \right)$ 
is smoothing. Indeed, its wave-front set is included 
$$
\left\{ (y,\eta,x,\xi)~|~ (y,\eta)\in \supp(\pi_i),~(x,\xi)\in \supp(1-\pi_0),
~(y,\eta)=\Phi_\tau^{T_i}(x,\xi)\,\right \},
$$
which is empty by construction. In particular, each $R_{i,\tau}$ is bounded from $H^{-1}(X)$ into $L^2(X).$ 
Moreover, a careful analysis of the construction of the Hadamard parametrix (as it is for instance presented 
in \cite{Ber})  shows that there exists a uniform constant $C$ such that :
\begin{equation}\label{UniformSmoothing}
\forall \,|\tau|\leq \tau_0,\,\forall w\in H^{-1}(X), \|R_i w\|_{L^2(X,g_0)}\,\leq C \| u\|_{H^{-1}(X,g_0)},
\end{equation} 
(see the appendix for a brief sketch of proof).

Since the $(\pi_i)_{i>1}$ are a partition of unity, $\sum_{i>1} \Pi_i\,=\, \mbox{Id}.$ We thus have 
\begin{equation}
\| u \|_0^2\,=\, \sum_{i=1}^N \langle \Pi_i u,u \rangle_0.
\end{equation}
Using that $u$ is an eigenfunction of $\Delta_\tau,$ we have 
\begin{eqnarray*}
\left | \langle \Pi_i u,u \rangle_0 \right |&=& \left | \langle \Pi_i U_\tau(T_i)u,u \rangle_0 \right | \\
 &\leq & \left | \langle \Pi_i U_\tau(T_i)\Pi_0 u,u \rangle_0 \right | + \left| \langle R_i u,u\rangle_0 \right|.
\end{eqnarray*}
We use Cauchy-Schwarz inequality on both terms and the following facts on the norm of the different operators :
\begin{itemize}
\item[-] $\Pi_i$ is bounded from $L^2(X)$ to itself, 
\item[-] $U_\tau(T)$ is an isometry from $L^2(X),$ 
\item[-] and \refeq{UniformSmoothing}.
\end{itemize}
This yields : 
$$
\left | \langle \Pi_i u,u \rangle_0 \right |\,\leq\,\left( c_i \| \Pi_0 u\|_0\,
+\, C\| u\|_{H^{-1}(X,g_0)}\right) \cdot \| u\|_0.
$$

Using that the $H^{-1}(X,g_\tau)$ norms are all equivalent and that 
$\| u\|_{H^{-1}(X,g_\tau)}^2\,= \,(1+E)^{-1} \| u\|_\tau^2$  when $u$ is an eigenfunction 
of $\Delta_\tau,$ we get 
$$
\left | \langle \Pi_i u,u \rangle_0 \right |\,\leq\, \left( c_i \| \Pi_0 u\| + C(1+E)^{-\und}\|u\|_0 \right )
\cdot  \|u\|_0,
$$
for a constant $C$ independent of $\tau\in [-\tau_0,\tau_0].$

Summing all these inequalities and denoting by ${\tilde{c}}$ the maximum of the $c_i$ we obtain:
$$
\| u\|_0 \leq (N\tilde{c})\| \Pi_0 u\| + CN ~(1+E)^{-\und}\| u\|_0.
$$
This yields the estimate in the proposition with $\Pi=\Pi_0$ and $c^{-1}\,=\,2\tilde{c}N$ 
as soon as $CN(1+E)^{-\und}$ is less than $\und.$ 
\end{proof}

This proposition has the following useful corollary. 
\begin{coro}\label{LowerCor}
Let $a_\tau$ be an analytically dependent family of non-negative symbols of order $0$ such that $\left \{ a_0(x,\xi) > \eps \right \}$ geometrically 
controls $(S^*X,g_0)$ then there exist positive $\kappa,\tau_0$ and $E_0$ such that, 
for any $|\tau|\leq \tau_0$ and any $u$ eigenfunction of $\Delta_\tau$ with energy greater 
than $E_0$ then 
$$ 
\langle A_\tau u,u\rangle_{0}\geq \kappa \|u\|_\tau^2,
$$
where $A_\tau$ denotes $\op(a_\tau).$
\end{coro}

\begin{remk}
In the following proof, we will use several times the fact that, for any symbols 
of order $0$, $\op(bc)-\op(b)\circ\op(c)$ is a pseudodifferential operator 
of order $-1.$ We will denote by $R$ any operator arising from this 
kind of operation. Observe that if $b$ and $c$ depend analytically on $\tau$, then 
$R$ also depends analytically on $\tau$.
\end{remk}

\begin{proof}
Let $\tau_0, \tilde{E}_0$ and $\Pi$ be given by applying Proposition \ref{UniformProp} to 
$\U \,=\,\left\{ a_0(x,\xi) > \eps \right\}.$ If needed, we restrict $\tau_0$ so that, 
for any $|\tau|\leq \tau_0$, $a_\tau$ is bounded below by $\frac{\varepsilon}{2}$ on $\U.$
Observe that 
\begin{eqnarray*}
\langle A_\tau u,u\rangle_0 & = & \langle A_\tau \Pi^2 u,u\rangle_0 + \langle A_\tau(\Id-\Pi^2)u,u\rangle_0 \\
& =& \langle \op(a_\tau \pi^2) u,u\rangle_0 + \langle \op(a_\tau (1-\pi^2)u,u\rangle_0 +\langle R_\tau u,u\rangle_0 \\
&\geq & \langle \op(a_\tau \pi^2) u,u\rangle_0 +\langle R_\tau u,u\rangle_0,
\end{eqnarray*}
where the last inequality holds because $\pi(x,\xi)\leq 1$ and we have chosen a positive quantization. 
Since the support of $\pi$ is included in $\U$ we have 
$$\forall (x,\xi)\in S^*X,~ a_\tau (x,\xi)\pi^2(x,\xi) \geq \frac{\eps}{2} \pi^2(x,\xi),$$ 
which is quantized into

\begin{eqnarray*}
\langle \op(a_\tau \pi^2)u,u\rangle_0 & \geq & \frac{\eps}{2} \langle \op(\pi^2) u,u \rangle_0 \\
	& \geq & \frac{\eps}{2} \| \Pi u\|_0^2 +\langle R_\tau u,u\rangle_0,
\end{eqnarray*}	
with another remainder operator $R_\tau$.
Putting these two inequalities together, we obtain 
\begin{equation}
\langle Au,u\rangle_0 \geq \frac{\eps}{2} \| \Pi u\|_0^2 +\langle R_\tau u,u\rangle_0.
\end{equation}
Since $R$ is of order $-1$ and the norms on $H^{-1}$ associated with $(g_\tau)_{|\tau|\leq \tau_0}$ 
are uniformly equivalent, there exists some uniform $C$ such that  
$$ \left|\langle R_\tau u,u\rangle_0 \right | \leq C E^{-\und} \|u\|_\tau.$$
Using Proposition \ref{UniformProp}, for $|\tau|\leq \tau_0$ and $E\geq E_0,$ we have 
$\| \Pi u\|_0 \geq c_1 \| \Pi u\|_\tau \geq c_2 \| u \|_\tau$ so that we get 
\begin{equation}
\langle A_\tau u,u\rangle \geq \left( c\frac{\eps}{2} -CE^{-\und}\right) \cdot \| u\|_\tau^2. 
\end{equation}
The corollary then holds with $\kappa = \frac{c\eps}{4}$ as soon as $E$ is large enough so that 
$CE^{-\und} \leq \displaystyle \frac{c\eps}{4}$. 
\end{proof}												
 
\subsection{Proof of Theorem \ref{IREthm}}
We consider normalized eigenbranches and we start from the formula 
for the derivative of the eigenfunctions : (in the following equation, everything 
except the constant $C$ depends on $\tau$) 
$$
\frac{d}{d\tau} \ln(E) \, \geq \, \langle A u,u\rangle\,-\,C E^{-\und}. 
$$
Under the assumptions of the theorem, we may use proposition \ref{UniformProp} 
and its corollary. This provides us with $\kappa>0$, $\tau_0$ and $E_0$ such that 
\begin{gather}
\nonumber \forall \, |\tau|\leq \tau_0,~~\forall \, E(\tau)\geq E_0,\\
\label{IncreasingBranches}\frac{d}{d\tau} \ln(E) \geq \kappa.
\end{gather}

The rest of the argument is merely reproduced from \cite{Hass}.\\
 
Denote by 
$R_M(E) \,=\,\displaystyle \int_{-\tau_0}^{\tau_0} R_M(t,E) \,dt.$ 
A straightforward manipulation of the integral yields : 
\begin{equation}\label{Manip}
R_M(E)\,=\, \sum_{j} \mbox{Leb}\{ \tau~|~E_j(\tau)\in (E-M,E+M]\,\},
\end{equation}
where $\mbox{Leb}$ denotes the Lebesgue measure.
We consider now $E$ such that $E-M\geq E_0.$ According to \refeq{IncreasingBranches}, 
any eigenbranch that crosses the energy $E-M$ has to do it in an increasing manner. 
This implies that any eigenbranch for which the set $I_j\,:=\,\{ t~|~E_j(\tau) \in (E-M,E+M]\,\}$ 
is not empty satisfies $E_j(-\tau_0)\leq E+M $ and $E_j(\tau_0)\geq E-M.$ In particular, 
the sum in \refeq{Manip} has only a finite number of positive terms, and moreover, 
this number is bounded by $N_{-\tau_0}(E+M)$  which is bounded by $CE^{\frac{d}{2}}$ 
according to Weyl's law. 
Estimate \refeq{IncreasingBranches} also implies that $I_j$ is an interval whose extremities 
we denote by $\tau_j^-$ and $\tau_j^+$ \footnote{ We do not specify whether 
$I_j$ is open or closed at each of its extremity since this fact depends on whether $\tau_j^{\pm}$  
coincide or not with $\pm\tau_0,$ and it is irrelevent when computing the Lebesgue measure.}.
Integrating \refeq{IncreasingBranches} we obtain 
\begin{eqnarray*}
\kappa (\tau_j^+-\tau_j^-) & \leq & \ln\left( \frac{E_j(\tau_j^+)}{E_j(\tau_j^-)}\right )\\
& \leq & \ln \left( \frac{E+M}{E-M}\right )\\ 
& \leq & \frac{C}{E},  
\end{eqnarray*}
as soon as $E$ is large enough (indepently of $j$). 
This bound each term of the sum in \refeq{Manip}, since the number of terms is also bounded explicitly 
using Weyl's law, we obtain 
$$
R_M(E) \leq C\cdot E^{\frac{d}{2}-1}.
$$
\hfill $\square$

\begin{remk}
The fact that $a_0$ is non-negative implies in particularly that the 
derivative of $vol_{g_\tau}X$ is increasing. In particular, the 
variation cannot be volume-preserving.
\end{remk}
  
\begin{remk}
Corollary \refeq{LowerCor} implies that for any semiclassical measure associated 
with $\Delta_\tau$, the following lower bound holds :
$$ \int a d\mu \geq \kappa.$$
If we replace the control assumption in Thm.\refeq{IREthm} by the following (weaker) assumption :
\begin{gather*} 
\forall \,\tau,~\forall \,\mu ~\mbox{semiclassical measure associated with}~\Delta_\tau,\\
\int a d\mu >0,
\end{gather*}
then, following the same lines of proof, we would get a 'almost-everywhere' result 
in the spirit of \cite{Hass}.
\end{remk}

\begin{remk} 
This result can probably be extended to manifolds with boundary with the precaution that, 
in this case we would have to control the possible escape of the 
mass to the boundary.
\end{remk}
\noindent {\bf Examples :}
\begin{enumerate}
\item A natural variation to look at is to simply take $g_\tau=\exp(\tau)g_0.$ 
In this case the control assumption is trivially satisfied so that we get 
a uniform integrated remainder. However, this is not a new result since 
looking at a fixed energy interval and dilating the volume is equivalent to 
fixing the volume and dilating the energy interval so that the 
integrated remainder is actually a regularized remainder 
in Weyl's law. And it is known that integrating the 
remainder in Weyl's law with respect to the energy yields a much better estimate than the 
universal remainder term.
\item For a more general conformal variation on a surface ($d=2$) we have  
$g_\tau:= (1+\tau f(x))^2g_0,$ the eigenvalue problem can be put in the form 
$$
\int_X \nabla_0 u \nabla_0 \phi\, dvol_{g_0} \,=\, E\int_X u\phi (1+\tau f(x))^2 \,dvol_{g_0}.
$$
This gives $a_0(x,\xi)=2f(x).$ Choosing $f$ non-negative and 
supported in a control region yields examples where the theorem applies.
\item
On the contrary, the variation considered in \cite{Hass} on a 
Donnelly surface doesn't satisfy the control estimate because 
of the cylinder of periodics orbits and of periodic geodesic 
in the hyperbolic wings. 
\end{enumerate}

\newpage

\section*{Appendix : Uniform Hadamard Parametrix}

The aim of this appendix is to give a (rather sketchy) justification of 
the uniform estimate \refeq{UniformSmoothing}. 

Let $X$ be a smooth manifold without boundary and $g_\tau$ a 
real-analytic family of riemannian metrics. 

Denote by $U_\tau(s)$ the propagator of the half-wave equation 
associated with $\Delta_\tau$ :
$$
U_\tau(s) \,=\, \exp(is\sqrt{\Delta_\tau}),
$$ 
and denote by $k_\tau(s,x,y)$ its Schwartz kernel. 
This kernel $k_\tau(s,x,y)$ is obtained by restricting 
the Schwartz kernel of $\cos(s\sqrt{\Delta_\tau})$ to 
the region $\sigma>0$ of the cotangent bundle $T^*(\R\times X\times Y)$

We start from the Hadamard parametrix construction as explained 
for instance in \cite{Ber}. Concretely, for $s$ less than some small $\eps$ and 
$x,y$ such that $d_\tau(x,y)\leq 2\eps$, we make the ansatz 
$$
{\tilde k}_N(\tau,s,x,y)\,=\,\sum_1^N u_k(\tau,x,y)T_{\alp_d+k}(s^2-d^2_\tau(x,y)),
$$
in which the $u_k$ are smooth functions of $x,y$ and $T_{\alp+k}(z)$ are oscillatory 
distributions of $\R$ with wave-front in $\{ (z=0,\zeta >0)\}.$ 
The coefficients $u_k$ satisfy some transport equations from 
which it is straightforward that they also depend smoothly on $\tau$. 
We fix a smooth cutoof function $\rho:=\rho(d_\tau(x,y))$ that cuts off in $d_\tau(x,y)\leq 2 \eps$ 
and we let $k_N=\tilde{k}_N\times \rho$ 
A computation then yields that $k_N$ solves the half-wave equation up to 
a remainder term $r_N(\tau,s,x,y)$ which is $\mathcal{C}^{l(N)}$ where $l(N)$ 
tends to infinity with $N.$ Applying Duhamel's 
principle, we get 
$$
k(\tau,s,x,y)-k_N(\tau,s,x,y)\,=\,\int_0^s k(\tau,s-s',x,w)r_N(\tau,s',w,y) ds' dw.
$$ 
Here again, a detailed analysis shows that, for $s\leq \eps$ and for any $N,$  
we finally have 
$$
k(\tau,s,x,y)\,=\, k_N(\tau,s,x,y)+R_N(\tau,s,x,y),
$$
in which everything is smooth in $\tau,$ and $R_N$ is ${\mathcal{C}}^{l(N)}.$ 
Since $k_N$ has a fixed order as a distribution of $(s,x,y),$ applying 
it to a $C^l$ function gives a $C^{l-\beta}$ with a loss $\beta$ that is independant 
of $N.$

We now fix some $T$, and we use the semigroup properties of $U$ to write 
$$
U_\tau(T)\,=\,U_{\tau}(\delta)\circ U_{\tau}(\delta)\circ\cdot\circ U_{\tau}(\delta),
$$ 
with a fixed small $\delta.$ We now plug into this expression the decomposition 
$U_\tau=K_N+R_N.$ Expanding everything gives a large sum in which each term 
is a composition of $K_N$ factors with $R_N$ factors. We address the regularity 
of each term. Each term has regularity  $N-p\beta$ where $p$ is the number of 
$K_N$ factors. Thus any term except the one with only $K_N$ factors may be 
made as regular as wanted by choosing $N$ large enough. Moreover, everything is 
regular in $\tau.$ 

It remains to address the term with only $K_N$ factors. Each $K_N$ is 
explicit and this composition is adressed using stationary phase. 
It is somewhat tedious in the general case because for instance 
of conjugate points. However it is much simpler when we have an initial and final 
microlocal cutoff that destroys everything in the wave-front set. Indeed 
a uniform non-stationary phase then applies that yields a uniform 
$C^m$ bound for any $m.$

\end{document}